\documentclass[11pt,letterpaper]{amsart}

\usepackage{amsfonts}
\usepackage{amsmath}
\usepackage{amssymb}
\usepackage{amsthm}
\usepackage[foot]{amsaddr}

\usepackage[margin=1in]{geometry}

\usepackage{microtype}

\usepackage{mathtools}
\mathtoolsset{centercolon}

\usepackage[hang,flushmargin]{footmisc}
\usepackage{enumitem}

\newtheorem{theorem}{Theorem}
\newtheorem{lemma}{Lemma}
\newtheorem{corollary}{Corollary}
\newtheorem{conjecture}[theorem]{Conjecture}

\usepackage[nobysame]{amsrefs}
\BibSpec{article}{%
    +{}  { \PrintAuthors}               {author}
    +{,} { \textit}                     {title}
    +{.} { }                            {part}
    +{:} { \textit}                     {subtitle}
    +{,} { \PrintContributions}         {contribution}
    +{.} { \PrintPartials}              {partial}
    +{,} { }                            {journal}
    +{,} { \voltext}                    {volume}
    +{,} { \issuetext}                  {number}
    +{,} { pp.~}                        {pages}
    +{}  { \PrintDatePV}                {date}
    +{,} { }                            {status}
    +{,} { \PrintDOI}                   {doi}
    +{,} { available at \eprint}        {eprint}
    +{}  { \parenthesize}               {language}
    +{}  { \PrintTranslation}           {translation}
    +{;} { \PrintReprint}               {reprint}
    +{.} { }                            {note}
    +{.} {}                             {transition}
    +{}  {\SentenceSpace \PrintReviews} {review}
}
\renewcommand{\eprint}[1]{\href{http://arxiv.org/abs/#1}{#1}}

\usepackage{array}
\usepackage{xspace}
\usepackage{bm}

\usepackage[colorlinks,citecolor=blue]{hyperref}

\DeclarePairedDelimiter\bra{\langle}{\rvert}
\DeclarePairedDelimiter\ket{\lvert}{\rangle}
\DeclarePairedDelimiterX\braket[2]{\langle}{\rangle}{#1 \delimsize\vert #2}

\newcommand{\floor}[1]{\lfloor{#1}\rfloor}

\DeclareMathOperator{\poly}{poly}

\DeclareMathOperator{\spn}{span}
\DeclareMathOperator{\tr}{tr}
\DeclareMathOperator{\Tr}{Tr}

\newcommand{\F}{\mathbb{F}}
\newcommand{\FF}{\mathbb{F}}
\newcommand{\CC}{\mathbb{C}}

\newcommand{\NN}{\mathbb{N}}
\renewcommand{\i}{\mathrm{i}}
\newcommand{\e}{\mathrm{e}}
\newcommand{\good}{\mathrm{good}}

\newcommand{\z}{Z}
\newcommand{\chr}{e}
\newcommand{\PGM}{\textsc{pgm}\xspace}
\newcommand{\HSP}{\textsc{hsp}\xspace}
\newcommand{\MAC}{\textsc{mac}\xspace}
\newcommand{\MACs}{\textsc{mac}s\xspace}

\newcommand{\eq}[1]{\hyperref[eq:#1]{(\ref*{eq:#1})}}
\renewcommand{\sec}[1]{\hyperref[sec:#1]{Section~\ref*{sec:#1}}}
\newcommand{\thm}[1]{\hyperref[thm:#1]{Theorem~\ref*{thm:#1}}}
\newcommand{\lem}[1]{\hyperref[lem:#1]{Lemma~\ref*{lem:#1}}}
\newcommand{\itm}[1]{\hyperref[itm:#1]{\ref*{itm:#1}}}
\newcommand{\app}[1]{\hyperref[app:#1]{Appendix~\ref*{app:#1}}}
\newcommand{\conj}[1]{\hyperref[conj:#1]{Conjecture~\ref*{conj:#1}}}

\newcommand{\tsum}{\textstyle\sum}

\renewcommand{\vec}[1]{{\bm{#1}}}
\newcommand{\cJ}{{\mathcal J}}
\newcommand{\cR}{{\mathcal R}}
\newcommand{\cV}{{\mathcal V}}

\title{Optimal quantum algorithm for polynomial interpolation}

\author{Andrew M.\ Childs$^{*,\dagger}$}
\author{Wim van Dam$^{\ddag}$}
\author{Shih-Han Hung$^{\dagger}$}
\author{Igor E.\ Shparlinski$^{\S}$}

\begin{document}

\renewcommand*{\thefootnote}{\fnsymbol{footnote}}
\footnotetext[1]{\textsc{Department of Computer Science and Institute for Advanced Computer Studies, University of Maryland}}
\footnotetext[2]{\textsc{Joint Center for Quantum Information and Computer Science, University of Maryland}}
\footnotetext[3]{\textsc{Departments of Computer Science and Physics, University of California, Santa Barbara}}
\footnotetext[4]{\textsc{Department of Pure Mathematics, University of New South Wales} \vspace{-2ex}}
\renewcommand*{\thefootnote}{\arabic{footnote}}


\begin{abstract}
We consider the number of quantum queries required to determine the coefficients of a degree-$d$ polynomial over $\F_q$.  A lower bound shown independently by Kane and Kutin and by Meyer and Pommersheim shows that $d/2+1/2$ quantum queries are needed to solve this problem with bounded error, whereas an algorithm of Boneh and Zhandry shows that $d$ quantum queries are sufficient.  We show that the lower bound is achievable: $d/2+1/2$ quantum queries suffice to determine the polynomial with bounded error.  Furthermore, we show that $d/2+1$ queries suffice to achieve probability approaching $1$ for large $q$.  These upper bounds improve results of Boneh and Zhandry on the insecurity of cryptographic protocols against quantum attacks.  We also show that our algorithm's success probability as a function of the number of queries is precisely optimal.  Furthermore, the algorithm can be implemented with gate complexity $\poly(\log q)$ with negligible decrease in the success probability.
We end with a conjecture about the quantum query complexity of multivariate polynomial interpolation.
\end{abstract}

\maketitle

\section{Introduction}
\label{sec:intro}

Let $f(X) = c_d X^d +\cdots + c_1 X + c_0\in\FF_q[X]$ be an unknown polynomial of degree $d$, specified by its coefficient vector $c \in \F_q^{d+1}$.
Suppose $q$ and $d$ are known\footnote{We assume $q>d$ so that different coefficients correspond to distinct functions $f\colon \FF_q \to \FF_q$.} and we are given a black box that evaluates $f$ on any desired $x \in \F_q$.  In the \emph{polynomial interpolation problem}, our goal is to learn $f$---that is, to determine the vector $c$---by querying this black box.  We would like to determine how many queries are required to solve this problem.

The classical query complexity of polynomial interpolation is well known: $d+1$ queries to $f$ are clearly sufficient and are also necessary to determine the polynomial, even with bounded error.  Shamir~\cite{Sha79} used this fact to construct a cryptographic protocol that divides a secret into $d+1$ parts such that knowledge of all the parts can be used to infer the secret, but any $d$ parts give no information about the secret. The security of this protocol relies on the fact that if $f$ is chosen uniformly at random, and if we only know $d$ function values $f(x_1),\dots,f(x_d)$, then we cannot guess the value $f(x_{d+1})$ for a point $x_{d+1} \notin \{x_1,\ldots,x_d\}$ with probability greater than $1/q$ (that is, there is no advantage over random guessing). This example motivates understanding the query complexity of polynomial interpolation precisely, since a single query can dramatically increase the amount of information that can be extracted.

The quantum query complexity of polynomial interpolation has also been studied previously.  Kane and Kutin \cite{KK11} and Meyer and Pommersheim \cite{MP11} independently showed that $d/2 + 1/2$ quantum queries are needed to solve the problem with bounded error.  Furthermore, Kane and Kutin conjectured that $d+1$ quantum queries might be necessary.  This was refuted by Boneh and Zhandry, who showed that $d$ quantum queries suffice to solve the problem with probability\footnote{While the notation $O(\cdot)$ only indicates an asymptotic upper bound on the absolute value, we sometimes write $1 - O(\cdot)$ to indicate a bound on a quantity that is at most $1$.} $1 - O(1/q)$ \cite{BZ13}.  To show this, they described a $1$-query quantum algorithm that determines a linear polynomial with probability $1 - O(1/q)$.  The result for general $d$ follows because $d-1$ classical queries can be used to reduce the case of a degree-$d$ polynomial to that of a linear polynomial.  However, this work left a substantial gap between the lower and upper bounds.

Here we present an improved quantum algorithm for polynomial interpolation.  We show that the aforementioned lower bounds are tight: with $d$ fixed, $k = d/2 + 1/2$ queries suffice to solve the problem with constant success probability.  While the success probability at this value of $k$ has a $q$-independent lower bound, it decreases rapidly with $k$, scaling like $1/k!$.  This raises the question of how the success probability increases as we make more queries.  We show that there is a sharp transition as $k$ is increased: in particular, with $k = d/2 + 1$ queries, the algorithm succeeds with a probability that approaches $1$ for large $q$.

Our algorithm is motivated by the pretty good measurement (\PGM) approach to the hidden subgroup problem (\HSP) \cite{BCD05}.  In this approach, one queries the black box on uniform superpositions to create \emph{coset states} and then makes entangled measurements on several coset states to infer the hidden subgroup.  As in the \PGM approach (and in other approaches to the \HSP using the so-called standard method), our algorithm makes nonadaptive queries to the black box and performs collective postprocessing.  Also, similarly to previous analysis of the \PGM approach, we can express our success probability in terms of the number of solutions of a system of polynomial equations.

However, our approach to polynomial interpolation also has significant differences from the \PGM approach to the \HSP. In particular, we introduce a different way to query the black box that simplifies both the algorithm and its analysis.  In the \PGM approach, we query the black box on a uniform superposition and then uncompute uniform superpositions over certain sets.  For polynomial interpolation, we instead query a carefully-chosen non-uniform superposition of inputs so that the subsequent uncomputation is classical.  Furthermore, the success probability of our method is higher, and its analysis is more straightforward, than if we used a direct analog of the \PGM approach.  We hope that these techniques will prove useful for other quantum algorithms, perhaps for the hidden subgroup problem or for other applications of the \PGM approach \cites{CD05,DDW07}.

We also show that our strategy is precisely optimal: for any number of queries $k$, we describe a $k$-query algorithm with the highest possible success probability.  We give a simple algebraic characterization of this success probability, as follows.

\begin{theorem}\label{thm:sucprob}
The maximum success probability of any $k$-query quantum algorithm for interpolating a polynomial of degree $d$ over $\F_q$ is $|R_k|/q^{d+1}$, where $R_k := \z(\F_q^k \times \F_q^k)$ is the range of the function $\z\colon \F_q^k \times \F_q^k \to \F_q^{d+1}$ defined by $\smash{\z(x,y)_j := \sum_{i=1}^k y_i x_i^j}$ for $j \in \{0,1,\ldots,d\}$.
\end{theorem}

We present an explicit quantum algorithm that achieves this success probability, and we show that no algorithm can do better. We establish optimality with an argument based on the dimension of the space spanned by the possible output states, which appears to be distinct from arguments using the two main approaches to proving limitations on quantum algorithms, the polynomial and adversary methods. Instead, our approach is closely related to a linear-algebraic lower bound technique of Radhakrishnan, Sen, and Venkatesh \cite{RSV02} and to the ``rank method'' of Boneh and Zhandry \cite{BZ13}.

We characterize the query complexity by proving bounds on $|R_k|$, as follows.

\begin{theorem}\label{thm:rangesize}
For any fixed positive integer $d$, the success probability of \thm{sucprob} is
\begin{enumerate}[itemsep=0pt,topsep=2pt,label=\textup{(\roman*)},ref=(\roman*)]
\item $|R_k|/q^{d+1} = \frac{1}{k!}(1 - O(1/q))$ if $d$ is odd and $k=\frac{d}{2}+\frac{1}{2}$, and \label{itm:rangesize_threshold}
\item $|R_k|/q^{d+1} = 1 - O(1/q)$ if $d$ is even and $k=\frac{d}{2} + 1$. \label{itm:rangesize_aboveth}
\end{enumerate}
\end{theorem}

To show the former bound, we explicitly characterize the possible $(x,y) \in \F_q^k \times \F_q^k$ such that $\z(x,y)$ takes a particular value.  We prove the latter bound in a completely different way, using a second moment argument.

\thm{rangesize} shows that the success probability has a sharp transition as a function of $k$, from subconstant for $k<d/2+1/2$ (by known lower bounds \cites{KK11,MP11}), to a ($d$-dependent) constant for $k=d/2+1/2$, to $1-o(1)$ for $k=d/2+1$.  Note that since $k$ must be an integer, the success probability varies differently with $k$ depending on whether $d$ is odd or even.  For fixed even $d$, $k=d/2+1$ queries give success probability $1-o(1)$, whereas $k=d/2$ queries give success probability $o(1)$.  For fixed odd $d$, the success probability is $o(1)$ for $k=d/2-1/2$ and constant for $k=d/2+1/2$.  To achieve higher success probability, we can make $k=d/2+3/2$ queries and treat $f$ as a polynomial of degree $d+1$ with $c_{d+1}=0$, giving success probability $1-o(1)$.

In light of these results, polynomial interpolation is reminiscent of the task of computing the parity of $n$ bits, where the classical query complexity is $n$ (even for bounded error) and the quantum query complexity is $n/2$ \cites{BBCMW01,FGGS98}.  More generally, a similar factor-of-two improvement is possible for the oracle interrogation problem, where the goal is to learn the entire $n$-bit string encoded by a black box \cite{Dam98}.  However, polynomial interpolation is qualitatively different in that the oracle returns values over $\F_q$ rather than $\F_2$.  Note that for the oracle interrogation problem over $\F_q$, one can only achieve speedup by a factor of about $1-1/q$ \cite{BZ13}*{Section 4}, which is negligible for large $q$.

Our algorithm improves results of Boneh and Zhandry giving quantum attacks on certain cryptographic protocols \cite{BZ13}.  For a version of the Shamir secret sharing scheme~\cite{Sha79} where the shares can be quantum superpositions, their $d$-query interpolation algorithm shows that a subset of only $d$ parties can recover the secret.  Our algorithm considerably strengthens this, showing that a subset of $d/2+1/2$ parties can recover the secret with constant probability, and $d/2+1$ can recover it with probability $1-O(1/q)$.  Boneh and Zhandry also formulate a model of quantum message-authentication codes (\MACs), where the goal is to tag messages to authenticate the sender.  Informally, a \MAC is called $d$-time if, given the ability to create $d$ valid message-tag pairs, an attacker cannot forge another valid message-tag pair.  Boneh and Zhandry show that there are $(d+1)$-wise independent functions that are not $d$-time quantum \MACs.  Our result improves this to show that there are $(d+1)$-wise independent functions that are not $(d/2+1/2)$-time quantum \MACs.

Finally, we consider the gate complexity of polynomial interpolation.  We call an algorithm \emph{gate-efficient} if it can be implemented with a number of $2$-qubit gates that is only larger than its query complexity by a factor of $\poly(\log q)$.  We construct a gate-efficient variant of our algorithm that achieves almost the same success probability.\footnote{Note that while our algorithm for $k=d/2+1/2$ has gate complexity polynomial in both $\log q$ and $d$, the algorithm for $k=d/2+1$ has gate complexity $k! \poly(\log q)$.  Improving the dependence on $d$ is a natural open question.}

\begin{theorem}\label{thm:efficiency}
For any fixed positive integer $d$, there is a gate-efficient quantum algorithm for interpolating a polynomial of degree $d$ over $\F_q$ using
\begin{enumerate}[itemsep=0pt,topsep=2pt,label=\textup{(\roman*)},ref=(\roman*)]
\item $k=\frac{d}{2}+\frac{1}{2}$ queries, succeeding with probability $\frac{1}{k!}(1-O(1/q))$, if $d$ is odd; and \label{itm:eff_threshold}
\item $k=\frac{d}{2}+1$ queries, succeeding with probability $1-o(1)$, if $d$ is even. \label{itm:eff_aboveth}
\end{enumerate}
\end{theorem}
The main step in implementing the algorithm is to invert the function $\z$ described in the statement of \thm{sucprob}, i.e., to find some $x,y \in \F_q^k$ so that $\z(x,y)$ takes a given value.  We achieve this by characterizing the solutions in terms of a polynomial equation and a system of linear equations.

In \sec{discussion} we discuss the more general case where $f\in\FF_q[X_1,\dots,X_n]$ is a multivariate polynomial of degree $d$. While our algorithm generalizes straightforwardly, the analysis of its success probability is more complicated. We conjecture that the quantum query complexity of this problem is smaller than the classical query complexity by a factor of $n+1$.

The remainder of the paper is organized as follows.
After introducing some definitions in \sec{prelim}, we describe our $k$-query algorithm in \sec{algdesc}.
We analyze the success probability of this algorithm for $k = d/2 + 1/2$ in \sec{at_threshold}, and for $k=d/2 + 1$ in \sec{above_threshold}.
We also show in \sec{pgm} that essentially the same performance can be achieved using $k$ independent queries to the oracle, each on a uniform superposition of inputs (which might make some cryptographic attacks easier, depending on the model).
We establish optimality of our algorithm in \sec{optimality}.
In \sec{efficiency}, we describe the gate-efficient version of our algorithm.
Finally, we conclude in \sec{discussion} with a brief discussion of some open questions.

\section{Quantum algorithm for polynomial interpolation}
\label{sec:alg}

\subsection{Preliminaries}
\label{sec:prelim}

Let $f(X) = c_d X^d + \cdots + c_1 X + c_0\in\FF_q[X]$ be an unknown polynomial of degree $d$ that is specified by the vector of coefficients $c\in\FF_q^{d+1}$, 
where $q=p^r$ a power of a prime $p$. Access to $f$ is provided by a black box acting as $\ket{x,y}\mapsto \ket{x,y+f(x)}$ for all $x,y\in\FF_q$.

Let $\chr\colon \FF_q\rightarrow \CC$ be the exponential function $\chr(z) = \e^{2\pi\i \Tr(z)/p}$, where the trace function $\Tr\colon \FF_q\rightarrow \FF_p$ is defined by $\Tr(z) = z+z^p + z^{p^2}+\cdots + z^{p^{r-1}}$.
The Fourier transform over $\FF_q$ is the unitary transformation acting as
$ \ket{x} \mapsto \frac{1}{\sqrt{q}}\sum_{y\in\FF_q} \chr(xy)\ket{y}
$
for all $x\in\FF_q$.

We can compute the value of $f$ into the phase by Fourier transforming the second query register.  If we apply the inverse Fourier transform, perform a query, and then apply the Fourier transform, we have the transformation
\begin{align}
	\ket{x,y}
	&\mapsto \frac{1}{\sqrt q} \sum_{z \in \F_q} \chr(-yz) \ket{x,z} \\
	&\mapsto \frac{1}{\sqrt q} \sum_{z \in \F_q} \chr(-yz) \ket{x,z+f(x)} \\
	&\mapsto \frac{1}{q} \sum_{z,w \in \F_q} \chr(-yz+(z+f(x))w) \ket{x,w} \\
	&= \chr(yf(x)) \ket{x,y}
\end{align}
for any $x,y \in \F_q$, where we used the fact that $\sum_{z \in \F_q} \chr(zv) = q \delta_{z,v}$.  We call the transformation $\ket{x,y} \mapsto \chr(yf(x))\ket{x,y}$ a \emph{phase query}. Since a phase query can be implemented with a single standard query and vice versa, the query complexity of a problem does not depend on which type of query we use.

For vectors $x,y\in\FF_q^k$, we denote the inner product over $\F_q$ by $x\cdot y := \sum_{i=1}^k x_i y_i$. The $k$-fold Fourier transform (i.e., the Fourier transform acting independently on each register) acts as
$  \ket{x}
  \mapsto \frac{1}{\sqrt{q^k}}\sum_{y\in\FF^k_q} \chr(x\cdot y)\ket{y}
$
for any $x \in \F_q^k$.

\subsection{The algorithm}
\label{sec:algdesc}

We now describe our algorithm for polynomial interpolation. An ideal algorithm would produce the Fourier transform of the coefficient vector $c \in \F_q^{d+1}$, that is, the state 
\begin{equation}
\ket{\hat{c}} = \frac{1}{\sqrt{q^{d+1}}}\sum_{z\in\FF_q^{d+1}}\chr(c\cdot z)\ket{z}. 
\end{equation}
Instead we use $k$ quantum queries to create the approximate state 
\begin{equation}
\ket{\hat{c}_{R_k}} := \frac{1}{\sqrt{|R_k|}}\sum_{z\in R_k}\chr(c\cdot z)\ket{z}
\end{equation}
for some set $R_k \subseteq \F_q^{d+1}$.
A measurement of this state in the Fourier basis gives $c$ with probability $|\braket{\hat{c}_{R_k}}{\hat{c}}|^2 = |R_k|/q^{d+1}$.

Our algorithm performs $k$ phase queries in parallel, each acting on a separate register.  On input $\ket{x,y}$ for $x,y \in \F_q^k$, these $k$ queries introduce the phase $\chr(\sum_{i=1}^k y_i f(x_i))$.  To define the set $R_k$, recall the function $\z\colon \FF_q^k\times\FF_q^k\rightarrow \FF_q^{d+1}$
defined by 
\begin{equation}
  \z(x,y)_j := \sum_{i=1}^k y_ix_i^j
  \text{~for~} j\in\{0,1,\dots,d\}.
\label{eq:z}
\end{equation}
Then we have
$\sum_{i=1}^k y_i f(x_i)
 = \sum_{i=1}^k \sum_{j=0}^d y_i c_j x_i^j
= c\cdot \z(x,y)$
for all $x,y\in\FF_q^k$.
The range $R_k := \z(\FF_q^k\times \FF_q^k)$ of the function $\z$ is the set 
\begin{equation}
  R_k = \{\z(x,y) : (x,y)\in \FF_q^k\times \FF_q^k\} \subseteq \FF_q^{d+1}.
\end{equation}
For each $z\in R_k$ we choose a unique $(x,y)\in\FF_q^k\times \FF_q^k$ such that $\z(x,y)=z$. 
Let $T_k \subseteq \FF^k_q\times \FF^k_q$ be the set of these representatives. Clearly, $\z\colon T_k \to R_k$ is a bijection.  

To create the state $\ket{\hat c_{R_k}}$, we prepare a uniform superposition over $T_k$, perform $k$ phase queries, and compute $\z$ in place (i.e., perform the unitary transformation $\ket{x,y} \mapsto \ket{\z(x,y)}$), giving
\begin{align}
\frac{1}{\sqrt{|T_k|}}\sum_{(x,y)\in T_k}\ket{x,y}
&\mapsto  \frac{1}{\sqrt{|T_k|}}\sum_{(x,y)\in T_k}\chr(c\cdot \z(x,y))\ket{x,y} \\
&\mapsto  \frac{1}{\sqrt{|R_k|}}\sum_{z\in R_k}\chr(c\cdot z)\ket{z}.
\end{align}

The above procedure is a $k$-query algorithm for polynomial interpolation that succeeds with probability $|R_k|/q^{d+1}$, establishing the lower bound on the success probability stated in \thm{sucprob}.  To analyze the algorithm, it remains to lower bound $|R_k|$ as a function of $k$.

\subsection{Performance using $d/2 + 1/2$ queries}
\label{sec:at_threshold}

We now consider the performance of the above algorithm using $k=d/2 + 1/2$ queries.  Let
\begin{align}
 \z^{-1}(z) = \{(x,y) \in \F_q^k \times \F_q^k : \z(x,y)=z\}
\end{align}
be the set of those $(x,y) \in \F_q^k \times \F_q^k$ corresponding to a particular $z \in \F_q^{d+1}$.  Clearly $|R_k|$ is the number of values of $z$ such that $Z^{-1}(z)$ is nonempty.  To analyze this, we focus on ``good'' values of $(x,y)$.  Define
$  X_k^\good := \{x \in \FF_q^k : x_i \ne x_j \; \forall \, i \ne j\}  \text{~and~}   
  Y_k^\good := (\FF_q^\times)^k$
and let $Z^{-1}(z)^\good := Z^{-1}(z) \cap (X_k^\good \times Y_k^\good)$.  We claim the following:

\begin{lemma}\label{lem:thresholdsize}
If $k = d/2 + 1/2$, then for all $z \in \FF_q^{d+1}$, either $|Z^{-1}(z)^\good|=0$ or $|Z^{-1}(z)^\good|=k!$.
\end{lemma}

\begin{proof}
We can write the condition $\z(x,y)=z$ in the form
$\sum_i y_i \vec x_i = z$, where $\vec x_i := (1,x_i,x_i^2,\ldots,x_i^d)$.
We claim that for a given $z \in \F_q^{d+1}$, the values $(x,y) \in X^\good \times Y^\good$ that satisfy this equation are unique up to a permutation of the indices.  To see this, suppose that $\z(x,y)=\z(u,v)$ for some good values $(x,y) \ne (u,v)$.  By permuting the indices, we can ensure that $x_i = u_i$ for $i \in \{1,\ldots,m\}$ and $x_i \ne u_i$ for $i \in \{m+1,\ldots,k\}$, where $m$ is the number of positions at which $x$ and $u$ agree.  Then we have
\begin{align}
	\sum_{i=1}^m (y_i - v_i) \vec x_i 
	+ \sum_{i=m+1}^k y_i \vec x_i 
	+ \sum_{i=m+1}^k v_i \vec u_i = 0.
\label{eq:lincomb}
\end{align}
It is well known that the Vandermonde matrix
\begin{align}
\begin{pmatrix}
1 & 1 & \cdots & 1 \\
x_1 & x_2 & \cdots & x_{d+1} \\
x^2_1 & x^2_2 & \cdots & x^2_{d+1} \\
\vdots & \vdots & & \vdots \\
x^d_1 & x^d_2 & \cdots & x^d_{d+1}
\end{pmatrix}
\end{align}
is invertible provided the values $x_1,x_2,\ldots,x_{d+1}$ are distinct.
Because the values $x_i$ for $i \in \{1,\ldots,k\}$ and $u_i$ for $i \in \{m+1,\ldots,k\}$ are all distinct, and because the number of terms in \eq{lincomb} is at most $2k \le d+1$, the vectors $\vec x_i$ for $i \in \{1,\ldots,k\}$ and $\vec u_i$ for $i \in \{m+1,\ldots,k\}$ are linearly independent.  Thus we have $y_i = v_i$ for all $i \in \{1,\ldots,m\}$ and $y_i = v_i = 0$ for all $i \in \{m+1,\ldots,k\}$.  Since $y \in Y^\good$, we cannot have $y_i=0$ for any $i$, so we must have $m=k$.  Therefore $x=u$ and $y=v$.  It follows that the only way to obtain a distinct $(x,y)$ is to permute the indices, and therefore we either have $|Z^{-1}(z)^\good|=0$ (if there is no $(x,y) \in X^\good \times Y^\good$ such that $\z(x,y)=z$) or $|Z^{-1}(z)^\good| = k!$.
\end{proof}

Using \lem{thresholdsize}, we can show that $k=d/2 + 1/2$ queries suffice to perform polynomial interpolation with probability that is independent of $q$, but that decreases with $d$.

\begin{proof}[Proof of \thm{rangesize}\itm{rangesize_threshold}: $k=d/2+1/2$]
We have $|X_k^\good| = q!/(q-k)!$ and $|Y_k^\good| = (q-1)^k$, so
\begin{align}
	\sum_{z \in \F_q^{d+1}} |Z^{-1}(z)^\good| 
	&= |X_k^\good| \cdot |Y_k^\good|
	= \frac{q!}{(q-k)!} (q-1)^k = q^{2k} (1 - O(1/q)).
	\label{eq:numgoodz}
\end{align}
Thus, invoking \lem{thresholdsize}, the number of values of $z$ for which $|Z^{-1}(z)^\good| = k!$ is at least $\frac{q^{2k}}{k!}(1-O(1/q))$.  Since $k=d/2 + 1/2$, it follows that $|R_k|/q^{d+1}$ is at least $\frac{1}{k!}(1-O(1/q))$, as claimed.
\end{proof}

\subsection{Performance using $d/2 + 1$ queries}
\label{sec:above_threshold}

Next we show that with more than $d/2 + 1/2$ queries, the success probability approaches $1$ for large $q$.

\begin{proof}[Proof of \thm{rangesize}\itm{rangesize_aboveth}: $k=d/2+1$]
Under the uniform distribution on $z \in \F_q^{d+1}$, we have
\begin{align}
 {|R_k|}/{q^{d+1}} = 1 - \Pr[|Z^{-1}(z)|=0]. 
\end{align}
 We use a second moment argument to upper bound the number of $z \in \F_q^{d+1}$ for which $|Z^{-1}(z)|=0$.  The mean of $|Z^{-1}(z)|$ is
 \begin{align}
  \mu := \frac{1}{q^{d+1}} \sum_{z \in \F_q^{d+1}} |Z^{-1}(z)| = q^{2k-(d+1)}.
\end{align}
Let $\delta[\mathcal P]$ be $1$ if $\mathcal P$ is true and $0$ if  
$\mathcal P$ is false.
For the second moment, we compute
\begin{align}
	\sum_{z \in \F_q^{d+1}} |Z^{-1}(z)|^2
	&= \sum_{u,v,x,y \in \F_q^k} \delta[\z(u,v)=\z(x,y)] \\
	&= \sum_{u,v,x,y \in \F_q^k} \frac{1}{q^{d+1}} \sum_{\lambda \in \F_q^{d+1}} \chr(\lambda \cdot (\z(u,v)-\z(x,y))) \\
	&= \frac{q^{4k}}{q^{d+1}} + \frac{1}{q^{d+1}} \sum_{\lambda \in \F_q^{d+1} \setminus (0,\ldots,0)}
	\Biggl( \sum_{x,y \in \F_q} \chr\biggl(y \sum_{j=0}^d \lambda_j  x^j\biggr) \Biggr)^{2k} \\
	&= q^{4k-(d+1)} + \frac{1}{q^{d+1}} \sum_{\lambda \in \F_q^{d+1} \setminus (0,\ldots,0)}
	\Biggl(q \sum_{x \in \F_q} \delta\Biggl[\sum_{j=0}^d \lambda_j x^j = 0\Biggr]\Biggr)^{2k} \\
	&\le q^{4k-(d+1)} + (qd)^{2k}.
\end{align}
Thus for the variance, we have
\begin{align}
\label{eq:variance}
	\sigma^2:= \frac{1}{q^{d+1}} \sum_{z \in \F_q^{d+1}} |Z^{-1}(z)|^2 - \mu^2 
	\le \frac{(qd)^{2k}}{q^{d+1}}.
\end{align}
(note that $\sigma^2 \ge 0$ by the Cauchy inequality). 
Applying the Chebyshev inequality, we find
\begin{align}
	\Pr[Z^{-1}(z)=0]
	\le \frac{\sigma^2}{\mu^2} 
	\le \frac{(qd)^{2k}/q^{d+1}}{q^{4k-2(d+1)}}
	= d^{2k} q^{d+1-2k}.
\end{align}
Therefore $|R_k|/q^{d+1} = 1-\Pr[Z^{-1}(z)=0] \ge 1 - d^{2k} q^{d+1-2k}$.  With $k=d/2 + 1$, we have 
\begin{align}
|R_k|/q^{d+1} &\ge 1 - d^{2k} / q = 1 - O(1/q)
\end{align}
as claimed.
\end{proof}

Note that one can improve the dependence on $d$ in \eq{variance} using results on the distribution of zeros in random polynomials \cite{KK90}.

\subsection{An alternative algorithm}
\label{sec:pgm}

The algorithm described above queries the oracle nonadaptively, that is, all $k$ queries can be performed in parallel.  However, the input state to these queries is correlated across all $k$ copies.  In this section, we describe an alternative algorithm that queries the black box on a state that is independent and identical for each of the $k$ queries, namely, a uniform superposition over all inputs.  This algorithm is suboptimal, but its performance is not significantly worse than that of the optimal algorithm described in \sec{algdesc}.

Analogous to the so-called standard method for the hidden subgroup problem, querying $f$ on a uniform superposition gives the state
$\frac{1}{\sqrt q} \sum_{x \in \F_q^k} \ket{x,f(x)}$.
If we use $k$ queries to prepare $k$ copies of this state and then perform the Fourier transform on the second register (or equivalently, perform $k$ independent phase queries), we obtain the state
\begin{align}
	\frac{1}{q^k} \sum_{x,y \in \F_q^k} \chr(c \cdot \z(x,y)) \ket{x,y}
	= \frac{1}{q^k} \sum_{z \in \F_q^{d+1}} \chr(c \cdot z) \sqrt{|Z^{-1}(z)|} \, \ket{Z^{-1}(z)}
\end{align}
where $\ket{Z^{-1}(z)} := \sum_{(x,y) \in Z^{-1}(z)} \ket{x,y}/|Z^{-1}(z)|^{1/2}$.
Motivated by the \PGM approach to the hidden subgroup problem \cite{BCD05}, suppose we perform the transformation $\ket{Z^{-1}(z)} \mapsto \ket{z}$, giving the state
\begin{align}
	\ket{\phi^c_k} := \frac{1}{q^k} \sum_{z \in \F_q^{d+1}} \chr(c \cdot z) \sqrt{|Z^{-1}(z)|} \, \ket{z}.
\end{align}
Measuring this state in the Fourier basis gives the outcome $c$ with probability 
\begin{align}
  |\braket{\phi^c_k}{\hat c}|^2
  = \frac{1}{q^{2k+d+1}} \bigg( \sum_{z \in \F_q^{d+1}} \sqrt{|Z^{-1}(z)|} \bigg)^2.
\end{align}

If $k=d/2+1/2$, we claim that this algorithm succeeds with constant probability.  From the proof of \thm{rangesize} for $k=d/2+1/2$, we have that $|Z^{-1}(z)| \ge k!$ for at least $\frac{q^{2k}}{k!}(1-O(1/q))$ values of $z$.  Therefore the success probability is at least $\frac{1}{k!} (1 - O(1/q))$.

If $k=d/2+1$, then this algorithm succeeds with probability that approaches $1$ for large $q$.  To see this, recall from the proof of \thm{rangesize} for $k=d/2+1$ that, under a uniform distribution over $z \in \F_q^{d+1}$, the quantity $Z^{-1}(z)$ has mean $\mu = q$ and standard deviation $\sigma = \sqrt{q} d^k$.  Thus, by the Chebyshev inequality, we have
\begin{align}
	\Pr\bigl[|Z^{-1}(z)| \le q - \alpha \sqrt{q} d^k\bigr] \le \frac{1}{\alpha^2}.
\label{eq:cheby_pgm}
\end{align}
It follows that
\begin{align}
	|\braket{\phi^c_k}{\hat c}|^2 \ge \biggl(1 - \frac{\alpha d^k}{\sqrt{q}}\biggr)\biggl(1-\frac{1}{\alpha^2}\biggr)^2.
\end{align}
Choosing $\alpha = \Theta(q^{1/6})$, this gives a success probability of $|\braket{\phi^c_k}{\hat c}|^2 = 1 - O(q^{-1/3})$, which approaches $1$ for large $q$.

\section{Optimality}
\label{sec:optimality}

In this section, we show that the query complexity of our algorithm is precisely optimal: no $k$-query algorithm can succeed with a probability larger than $|R_k|/q^{d+1}$.  We begin with a basic result showing that $m$ states spanning an $n$-dimensional subspace can be distinguished with probability at most $n/m$.

\begin{lemma}\label{lem:prob_dim}
Suppose we are given a state $\ket{\psi_c}$ with $c \in C$ chosen uniformly at random.  Then the probability of correctly determining $c$ with some orthogonal measurement is at most $\dim\spn\{ \ket{\psi_c} : c \in C \}/|C|$.
\end{lemma}

\begin{proof}
Consider a measurement with orthogonal projectors $E_c$, and let $\Pi$ denote the projection onto $\spn\{\ket{\psi_c} : c \in C\}$.  Then we have
\begin{align}
	\Pr[\text{success}]
	&= \frac{1}{|C|} \sum_{c \in C} \bra{\psi_c} E_c \ket{\psi_c}
	\le \frac{1}{|C|} \sum_{c \in C} \tr(E_c \Pi)
	= \frac{\tr(\Pi)}{|C|}
	= \frac{\dim\spn\{ \ket{\psi_c} : c \in C \}}{|C|}
\end{align}
as claimed.
\end{proof}

We apply this lemma where $\ket{\psi_c}$ is the final state of a given quantum query algorithm when the black box contains $c \in \F_q^{d+1}$.  There is no loss of generality in considering an orthogonal measurement at the end of the algorithm since we allow the use of an arbitrary-sized ancilla.

\begin{lemma}\label{lem:dim_range}
Let $\ket{\psi_c}$ be the state of any quantum polynomial interpolation algorithm after $k$ queries, where the black box contains $c \in \F_q^{d+1}$.  Then $\dim\spn\{\ket{\psi_c} : c \in \F_q^{d+1}\} \le |R_k|$.
\end{lemma}

\begin{proof}
We claim that
\begin{align}
	\ket{\psi_c} = \sum_{x,y \in \F_q^k} \chr(\z(x,y) \cdot c) \ket{\phi_{x,y}}
\end{align}
for some set of (unnormalized) states $\{\ket{\phi_{x,y}} : x,y \in \F_q^k\}$ that do not depend on $c$.  Then the result follows, since
\begin{align}
	\ket{\psi_c} 
	= \sum_{z \in \F_q^{d+1}} \chr(z \cdot c) \sum_{x,y \in Z^{-1}(z)} \ket{\phi_{x,y}}
	\in \spn\biggl\{\sum_{x,y \in Z^{-1}(z)} \ket{\phi_{x,y}} : z \in \F_q^{d+1}\biggr\},
\end{align}
which has dimension at most $|R_k| = |\{Z(x,y):(x,y)\in\F_q^k\times\F_q^k\}|$.

To see the claim, consider a general $k$-query algorithm $U_k Q_c U_{k-1} \ldots Q_c U_1 Q_c U_0$ acting on states of the form $\ket{x,y,w}$ for an arbitrary-sized workspace register $\ket{w}$, starting in the state $\ket{x_0,y_0,w_0}=\ket{0,0,0}$.  Here $Q_c\colon \ket{x,y,w} \mapsto \chr(y f(x)) \ket{x,y,w}$ is a phase query.  The final state $\ket{\psi_c}$ equals
\begin{align}
	\sum_{\substack{x,y \in \F^{k}_q \\
        x_{k+1}, y_{k+1}\in\F_q\\
	w\in I^{k+1}}}
	\chr\biggl(\sum_{j=1}^k y_j f(x_j)\biggr)
	\biggl( \prod_{j=0}^k \bra{x_{j+1},y_{j+1},w_{j+1}} U_j \ket{x_j,y_j,w_j} \, \biggr)
	\ket{x_{k+1},y_{k+1},w_{k+1}},
\end{align}
with $x_0=y_0=w_0=0$, $x=(x_1,\ldots,x_k)$, $y=(y_1,\ldots,y_k)$, $w=(w_1,\dots,w_{k+1})$, and $I$ some appropriate index set.
This expression has the claimed form when we define
\begin{align}
	\ket{\phi_{x,y}} &= \sum_{\substack{x_{k+1},y_{k+1} \in \F_q \\ w\in I^{k+1}}}
	\biggl( \prod_{j=0}^k \bra{x_{j+1},y_{j+1},w_{j+1}} U_j \ket{x_j,y_j,w_j} \biggr)
	\ket{x_{k+1},y_{k+1},w_{k+1}}. \qedhere
\end{align}
\end{proof}

We can now prove our upper bound on the success probability of quantum algorithms for polynomial interpolation.

\begin{proof}[Proof of \thm{sucprob} (upper bound on success probability)]
By combining \lem{prob_dim} with \lem{dim_range}, we see that if the coefficients $c \in \F_q^{d+1}$ are chosen uniformly at random, no algorithm can succeed with probability greater than $|R_k|/q^{d+1}$.  Since the minimum cannot be larger than the average, this implies a lower bound on the success probability in the worst case of $|R_k|/q^{d+1}$.
\end{proof}

This result also shows that the exact quantum query complexity of polynomial interpolation is maximal.

\begin{corollary}\label{cor:exact}
The exact quantum query complexity of interpolating a degree-$d$ polynomial is $d+1$.
\end{corollary}

\begin{proof}
This follows from \thm{sucprob} and the fact that if $k<d+1$, we have $|R_k|<q^{d+1}$.  To see this, observe that if $k<d+1$, then vectors of the form $(0,\ldots,0,z_d)$ for $z_d \ne 0$ are not in the range of $\z$.  We can assume there is an $(x,y) \in Z^{-1}(z)$ with $x_1,\ldots,x_k$ all distinct, since if $x_i = x_j$ for some $i \ne j$, then we could delete index $j$ and replace $y_i$ by $y_i + y_j$.  Then in equation \eq{yfromxandz}, the Vandermonde matrix on the left-hand side is invertible, so $y_1 = \cdots = y_k = 0$.  However, this implies that $\sum_i y_i x_i^d = 0 \ne z_d$.
\end{proof}

\section{Gate complexity}
\label{sec:efficiency}

In \sec{alg}, we analyzed the query complexity of our polynomial interpolation algorithm.  Here we describe a $(d/2 + 1/2)$-query algorithm whose gate complexity is $\poly(\log q)$, and whose success probability is close to that of the best algorithm using this number of queries (in particular, for fixed $d$ it still succeeds with constant probability).
We also give an algorithm for the case $k=d/2 + 1$ whose gate complexity is larger by a factor of $\poly(\log q)$, but with an additional factor of $k!$.

\subsection{Algorithm for $k = d/2 + 1/2$ queries}
\label{sec:efficiency_at_threshold}

To simplify the computation of unique representatives of values $z \in R_k$, we restrict attention to the ``good'' case considered in \sec{at_threshold}.
Let
\begin{align}
	R_k^\good := \{\z(x,y) : x \in X_k^\good,\, y \in Y_k^\good\}.
\end{align}
For any $z \in R_k^\good$, we show how to efficiently compute representative values $x \in X_k^\good$ and $y \in Y_k^\good$ with $\z(x,y)=z$, defining a set of representatives $T_k^\good$.
Then we consider an algorithm as described in \sec{algdesc}, but with $R_k$ replaced by $R_k^\good$ and $T_k$ replaced by $T_k^\good$.  Clearly the success probability of this algorithm is $|R_k^\good|/q^{d+1}$.
Our lower bound on $|R_k|$ in \sec{at_threshold} was actually a bound on $|R_k^\good|$, so this algorithm still succeeds with probability $\frac{1}{k!}(1 + O(1/q))$.

To give a gate-efficient algorithm, it suffices to show how to efficiently compute the function $\z^{-1}\colon R_k^\good \to T_k^\good$ (that is, to compute this function using $\poly(\log q)$ gates).

\begin{lemma}\label{lem:inverse}
Suppose there is an efficient algorithm to compute $\smash{\z^{-1}\colon R_k^\good \to T_k^\good}$.  Then the algorithm of \sec{algdesc} can be made gate-efficient (with $R_k$ replaced by $R_k^\good$ and $T_k$ by $T_k^\good$).
\end{lemma}

\begin{proof}
It is trivial to compute $\z\colon T_k^\good \to R_k^\good$ efficiently.  Given an efficient procedure for computing $\z^{-1}\colon R_k^\good \to T_k^\good$, this gives us the ability to efficiently compute $\z$ in place (that is, to perform the transformation $\ket{x,y} \mapsto \ket{z}$ as required by the algorithm).
To do this, we first compute $z$ in an ancilla register by evaluating $\z$ (which only requires arithmetic over $\F_q$) and then uncompute $(x,y)$ by applying the circuit for $\z^{-1}$ in reverse.

It remains to prepare the initial uniform superposition over $T_k^\good$.  This can also be done using the ability to compute $\z^{-1}$.  Suppose we create a uniform superposition over all of $z \in \F_q^{d+1}$ and then attempt to compute $\z^{-1}$.  If $z \notin \smash{R_k^\good}$, this is detected, and we can set a flag qubit indicating failure.  Thus we can prepare a state of the form
\begin{align}
	\frac{1}{\sqrt{q^{d+1}}}
	\Biggl(\sum_{(x,y) \in T_k^\good} \ket{\z(x,y),0,x,y} 
	+ \sum_{z \in \F_q^{d+1}\setminus R_k^\good} \ket{z,1,0,0}\Biggr).
\end{align}
A measurement of the flag qubit gives the outcome $0$ with probability $|R_k^\good|/q^{d+1}$.  Since this is our lower bound on the success probability of the overall algorithm, we do not have to repeat this process too many times before we successfully prepare the initial state (and by sufficiently many repetitions, we can make the error probability arbitrarily small).  When the measurement succeeds, we can uncompute the first register to obtain the state $\sum_{(x,y) \in T_k^\good} \ket{x,y}/|T_k^\good|^{1/2}$ as desired.
\end{proof}

In the remainder of this section, we describe how to efficiently compute $\z^{-1}(z)$ for $z \in R_k^\good$.
Our approach appeals to ``Prony's method''~\cite{PCM88} (a precursor to Fourier analysis) and the theory of linear recurrences.
We start with the following technical result, where $e_j$ denotes the $j$th elementary symmetric polynomial in $k$ variables, i.e.,
\begin{align}
  e_j(x_1,\dots,x_k) 
  = \sum_{1\leq i_1 < i_2 <\cdots < i_j \leq k} x_{i_1} x_{i_2}\cdots x_{i_j}. 
\end{align}

\begin{lemma}\label{lem:sympoly}
We have
\begin{align}
  x_i^k = - \sum_{j=1}^k x_i^{k-j} (-1)^j e_j(x_1,\ldots,x_k)
\label{eq:sympoly}
\end{align}
for all $i \in \{1,\ldots,k\}$.
\end{lemma}

\begin{proof}
Observe that it suffices to prove the lemma for $i=1$, since if we interchange the roles of $x_1$ and $x_\ell$ in \eq{sympoly} with $i=1$, we obtain \eq{sympoly} with $i=\ell$.
 
We apply induction on $k$.  If $k=1$ then the claim is trivial.  Now suppose the claim holds for a given value of $k$.
We have
\begin{align}
   e_j(x_1,\ldots,x_{k+1})
   &= e_j(x_1,\ldots,x_k) + x_{k+1} e_{j-1}(x_1,\ldots,x_k) \\
   &= e_j(x_1,\ldots,x_k) + x_{k+1} \tfrac{\partial}{\partial{x}_{k+1}} e_j(x_1,\ldots,x_{k+1}).
\end{align}
Therefore
\begin{align}
   -\sum_{j=1}^{k+1} x_1^{k+1-j} (-1)^j e_j(x_1,\ldots,x_{k+1})
   &= -\sum_{j=1}^{k+1} x_1^{k+1-j} (-1)^j \bigl[e_j(x_1,\ldots,x_k)\\
   &\qquad\quad + x_{k+1} \tfrac{\partial}{\partial{x}_{k+1}} e_j(x_1,\ldots,x_{k+1})\bigr] \nonumber \\
   &= x_1^{k+1} - x_{k+1} \tfrac{\partial}{\partial{x}_{k+1}} x_i^{k+1} \\
   &= x_1^{k+1}
\end{align}
(where the second equality uses the induction hypothesis).  
\end{proof}

Using this fact, we can show that each component of $\z(x,y)$ satisfies a $k$th-order linear recurrence.

\begin{lemma}\label{lem:recurrence}
If $z_j = \sum_{i=1}^k y_i x_i^j$ for all nonnegative integers $j$, then we have (for all nonnegative integers $n$)
\begin{align}
  z_{n+k} = -\sum_{j=0}^{k-1} (-1)^{k-j} e_{k-j}(x_1,\ldots,x_k) z_{n+j}.
\label{eq:recurrence}
\end{align}
\end{lemma}
\begin{proof}
The right-hand side of~\eq{recurrence} is
\begin{align}
 -\sum_{j=0}^{k-1} (-1)^{k-j} e_{k-j}(x_1,\ldots,x_k) \sum_{i=1}^k y_i x_i^{n+j}
  &= -\sum_{i=1}^k y_i \sum_{j=0}^{k-1} (-1)^{k-j} e_{k-j}(x_1,\ldots,x_k) x_i^{n+j} \\
  &= -\sum_{i=1}^k y_i \sum_{j=1}^k (-1)^j e_j(x_1,\ldots,x_k) x_i^{n+k-j} \\
  &= \sum_{i=1}^k y_i x_i^{n+k} =  z_{n+k}
\end{align}
as claimed, where the third equality uses \lem{sympoly}.
\end{proof}

We are now ready to describe the gate-efficient algorithm for polynomial interpolation.

\begin{proof}[Proof of \thm{efficiency}\itm{eff_threshold}: $k=d/2+1/2$]
By \lem{inverse}, it suffices to give an efficient algorithm for computing a representative $(x,y) \in Z^{-1}(z)^\good$ for any given $z \in R_k^\good$.

By \lem{recurrence}, the coefficients $a_j = -(-1)^{k-j} e_{k-j}(x_1,\ldots,x_k)$ of the linear recurrence \eq{recurrence} satisfy
\begin{align}
	H_k {\begin{pmatrix} a_0 \\ a_1 \\ \vdots \\ a_{k-1} \end{pmatrix}}
	= {\begin{pmatrix} z_k \\ z_{k+1} \\ \vdots \\ z_{2k-1} \end{pmatrix}}, 
\label{eq:hankeleq}
\quad\text{where}\quad
	H_k &:= {\begin{pmatrix}
  z_0 & z_1 & \cdots & z_{k-1} \\
  z_1 & z_2 & \cdots & z_k \\
  \vdots & \vdots & \ddots & \vdots\\
  z_{k-1} & z_k & \dots & z_{2(k-1)}
  \end{pmatrix}}
\end{align}
is a Hankel matrix.
Observe that
\begin{align}
H_k &=
V_k^T
{\begin{pmatrix}
y_1 & 0 & \cdots & 0 \\
0 & y_2 & & 0\\
\vdots & & \ddots & \vdots \\
0 & 0 & \cdots &  y_k
\end{pmatrix}}
V_k,
\quad\text{where}\quad
V_k := {\begin{pmatrix}
1 & x_1 & x_1^2 & \cdots & x_1^{k-1} \\
1 & x_2 & x_2^2 & \cdots & x_2^{k-1} \\
\vdots & \vdots & \vdots & & \vdots\\
1 & x_k & x_k^2 & \cdots & x_k^{k-1}
\end{pmatrix}}.
\end{align}
For $x \in X_k^\good$, the Vandermonde matrix $V_k$ (and its transpose) are invertible, and for $y \in Y_k^\good$, the diagonal matrix is invertible.  Then $H_k$ is invertible, and we have
\begin{align}
	{\begin{pmatrix} a_0 \\ a_1 \\ \vdots \\ a_{k-1} \end{pmatrix}}
	= H_k^{-1} {\begin{pmatrix} z_k \\ z_{k+1} \\ \vdots \\ z_{2k-1} \end{pmatrix}}.
\label{eq:polycoeffs}
\end{align}

We claim that for any $(x,y) \in Z^{-1}(z)^\good$, the values $x_1,\ldots,x_k$ must be roots of the characteristic polynomial
\begin{align}
	\chi(x) := x^k - \sum_{j=0}^{k-1} a_j x^j.
\label{eq:charpoly}
\end{align}
To see this, observe that
\begin{align}
	{\begin{pmatrix} z_k \\ z_{k+1} \\ \vdots \\ z_{2k-1} \end{pmatrix}} -
	H_k {\begin{pmatrix} a_0 \\ a_1 \\ \vdots \\ a_{k-1} \end{pmatrix}}
 = V_k^T
{\begin{pmatrix}\chi(x_1) \\ \chi(x_2) \\ \vdots \\ \chi(x_k)\end{pmatrix}}.
\end{align}
This must be the zero vector, and since the Vandermonde matrix is invertible, we see that $\chi(x_i) = 0$ for all $i \in \{1,\ldots,k\}$.\footnote{Since a polynomial of degree $k$ can have at most $k$ roots, this shows that the values $x_1,\ldots,x_k$ are unique up to permutation, giving $|{Z^{-1}(z)^\good}|=k!$ as shown in \lem{thresholdsize}.}

Finally, observe that the values $y_1,\ldots,y_k$ satisfy
\begin{align}
V_k^T
\begin{pmatrix}
y_1  \\ \vdots \\ y_k
\end{pmatrix}
 = 
\begin{pmatrix}
z_0 \\ z_1 \\ \vdots\\ z_{k-1}
\end{pmatrix}.
 \label{eq:yfromxandz}
\end{align}
Since the Vandermonde matrix is invertible, we see that $y$ is uniquely determined by $z$ and $x$.

To compute a unique representative of a given $z \in R_k^\good$, we use equation~\eq{polycoeffs} to efficiently compute the coefficients $a_0,\ldots,a_{k-1}$ of the characteristic polynomial $\chi(x)$.  We can then determine $x \in X_k^\good$ by finding the roots of this polynomial, which can be done in time $\poly(k, \log q)$ using a randomized algorithm \cite[Chapter 14]{GG13}.  Finally, we can determine $y \in Y_k^\good$ by solving a linear system of equations, namely~\eq{yfromxandz}.

This procedure does not uniquely specify $(x,y)$ because any permutation of the indices (acting identically on $x$ and $y$) gives an equivalent solution.  To choose a unique $(x,y) \in T_k^\good$, we simply require that the entries of $x$ occur in lexicographic order with respect to some fixed representation of $\F_q$.
\end{proof}

\subsection{Algorithm for $k = d/2 + 1$ queries}

We now present a similar algorithm for the case $k=d/2+1$ that also has gate complexity $\poly(\log q)$, although it has more overhead as a function of $d$.

To apply the approach of \sec{efficiency_at_threshold}, we again focus on solutions of $\z(x,y)=z$ with $(x,y) \in X^\good \times Y^\good$.  However, recall that our lower bound on the success probability for $k=d/2+1$ in \sec{above_threshold} used all solutions $(x,y) \in \F_q^k \times \F_q^k$.  Thus we begin by showing that the success probability of the algorithm remains close to $1$ even when restricted to good solutions.

\begin{lemma}\label{lem:rgood_above_threshold}
If $k = d/2+1$, then $|R_k^\good|/q^{d+1} = 1-O(1/q)$.
\end{lemma}

\begin{proof}
We repeat the second moment argument of \sec{above_threshold}, but now restricted to good solutions.  Under the uniform distribution on $z \in \F_q^{d+1}$, we have
\begin{align}
	\mu^\good
	:= \frac{1}{q^{d+1}} \sum_{z \in \F_q^{d+1}} |{Z^{-1}(z)^\good}|
	= q^{2k-(d+1)} (1 - O(1/q))
\label{eq:atg_mean}
\end{align}
by \eq{numgoodz}.  Similarly to the previous second moment calculation, we have
\begin{align}
	\sum_{z \in \F_q^{d+1}} \! |{Z^{-1}(z)^\good}|^2
	&= \sum_{u,x \in X_k^\good} \sum_{v,y \in Y_k^\good} \!\! \delta[\z(u,v)=\z(x,y)] \\
	&= q^{d+1} (\mu^\good)^2 + \frac{1}{q^{d+1}} \sum_{u,x \in X_k^\good} \sum_{v,y \in Y_k^\good} \sum_{\lambda \in \F_q^{d+1}} \!\! \chr(\lambda \cdot (\z(u,v)-\z(x,y))).
\end{align}
Thus we have
\begin{align}
	&(\sigma^\good)^2
	:= \frac{1}{q^{d+1}} \sum_{z \in \F_q^{d+1}} |{Z^{-1}(z)^\good}|^2 - (\mu^\good)^2 \\
	&\quad = \frac{1}{q^{2(d+1)}} \sum_{\lambda \in \F_q^{d+1}} \sum_{u,x \in X_k^\good} \sum_{v,y \in Y_k^\good} \prod_{i=1}^k \chr(v_i \tsum_{j=0}^d \lambda_j u_i^j) \, \chr(-y_i \tsum_{j=0}^d \lambda_j x_i^j) \\
	&\quad= \frac{1}{q^{2(d+1)}} \sum_{\lambda \in \F_q^{d+1}} \sum_{u,x \in X_k^\good} \prod_{i=1}^k (q \, \delta[\tsum_{j=0}^d \lambda_j u_i^j=0] - 1) (q \, \delta[\tsum_{j=0}^d \lambda_j x_i^j = 0] - 1) \\
	&\quad\le \frac{1}{q^{2(d+1)}} \sum_{\lambda \in \F_q^{d+1}} \sum_{u,x \in \F_q^k} \prod_{i=1}^k (q \, \delta[\tsum_{j=0}^d \lambda_j u_i^j=0] + 1) (q \, \delta[\tsum_{j=0}^d \lambda_j x_i^j = 0] + 1) \\
	&\quad\le \frac{(q(d+1))^{2k}}{q^{d+1}}
	\label{eq:atg_variance}
\end{align}
(which is identical to the previous bound for $\sigma^2$ except that $d$ is replaced by $d+1$).
Therefore, by the Chebyshev inequality, we have
\begin{align}
	\Pr[Z^{-1}(z)^\good]
	\le \frac{(\sigma^\good)^2}{(\mu^\good)^2}
	\le (d+1)^{2k} q^{d+1-2k}.
\end{align}
With $k=d/2+1$, we find
\begin{align}
	\frac{|R_k^\good|}{q^{d+1}}
	= 1 - \Pr[Z^{-1}(z)^\good = 0] \ge 1 - \frac{(d+1)^{2k}}{q} = 1 - O(1/q)
\end{align}
as claimed.
\end{proof}

Now consider the problem of computing a value $(x,y) \in X^\good \times Y^\good$ such that $\z(x,y)=z$ for some given $z \in R_k^\good$.  We can approach this task using the strategy outlined in \sec{efficiency_at_threshold}.  With $k=d/2+1$, we have $2k-1 = d+1$, so the last entry in the vector on the right-hand side of \eq{polycoeffs} is not specified.  Nevertheless, for any fixed $(x,y) \in \F_q^k \times \F_q^k$, the value $z_{d+1} = \z(x,y)_{d+1}$ is well-defined by extending \eq{z} to $j=d+1$, so we can find $(x,y) \in X^\good \times Y^\good$ by searching for some value of $z_{d+1} \in \F_q$ for which the algorithm of \sec{efficiency_at_threshold} succeeds at finding $k$ distinct roots $x_1,\ldots,x_k \in \F_q$ of the characteristic polynomial \eq{charpoly}.

We claim that choosing a random $z_{d+1} \in \F_q$ gives a solution with probability nearly $1/k!$.

\begin{lemma}\label{lem:random_extra_z}
Suppose $z=(z_0,\ldots,z_d)$ is chosen uniformly at random from $\F_q^{d+1}$.  Then with probability $1-o(1)$ (over the choice of $z$), choosing $z_{d+1}$ uniformly at random from $\F_q$ and solving for $(x,y) \in Z^{-1}(z)^\good$ as in the proof of \thm{efficiency}\itm{eff_aboveth} gives a solution with probability $(1-o(1))/k!$ (over the choice of $z_{d+1}$).
\end{lemma}

\begin{proof}
For any $z \in R_k^\good$, each value of $z_{d+1} \in \F_q$ gives a unique set of roots of the characteristic polynomial \eq{charpoly}, and hence corresponds to either $0$ or $k!$ solutions $(x,y) \in Z^{-1}(z)^\good$.
By a similar second moment argument as in \eq{cheby_pgm}, but using the mean \eq{atg_mean} and variance \eq{atg_variance} of $Z^{-1}(z)^\good$, we have $|Z^{-1}(z)^\good| = q(1-o(1))$ with probability $1-o(1)$ over the uniform choice of $z \in \F_q^{d+1}$.
Thus the number of values of $z_{d+1}$ that lead to a valid solution $(x,y) \in Z^{-1}(z)^\good$ is at least $q(1-o(1))/k!$ with probability $1-o(1)$ over the choice of $z$.  Since there are $q$ possible values of $z_{d+1}$, choosing $z_{d+1}$ at random leads to a valid representative $(x,y) \in Z^{-1}(z)^\good$ with probability $(1-o(1))/k!$, again with probability $1-o(1)$ over the uniform choice of $z$.
\end{proof}

\lem{random_extra_z} gives a method for computing a representative $(x,y) \in X^\good \times Y^\good$ such that $\z(x,y)=z$: simply choose $z_{d+1} \in \F_q$ at random until we find a solution.  Repeating this process $O(k!)$ times suffices to find a solution with constant probability (for almost all $z$).  However, since this approach constructs a random $(x,y) \in Z^{-1}(z)^\good$ rather than a unique representative, it does not define a set $T_k^\good$, and it cannot be directly applied to our quantum algorithm as described so far.
Instead, we construct an equivalent algorithm that represents the sets $Z^{-1}(z)^\good$ using quantum superpositions.

\begin{lemma}\label{lem:super_rep}
Suppose there is an efficient algorithm to generate the quantum state 
\begin{align}
  \ket{Z^{-1}(z)^\good}
  := \frac{1}{\sqrt{|Z^{-1}(z)^\good|}} \sum_{(x,y) \in Z^{-1}(z)^\good} \ket{x,y}
\end{align}
for any given $z \in R_k^\good$.  Then there is a gate-efficient $k$-query quantum algorithm for the polynomial interpolation problem, succeeding with probability $|{R_k^\good}|/q^{d+1}$.
\end{lemma}

\begin{proof}
We essentially replace $(x,y) \in T_k^\good$ by $\ket{Z^{-1}(\z(x,y))^\good}$ throughout the algorithm.  More concretely, we proceed as follows.

Observe that the ability to perform the given state generation map $\ket{z} \mapsto \ket{z}\ket{Z^{-1}(z)^\good}$ implies the ability to perform the in-place transformation
\begin{align}
	\ket{z} \mapsto \ket{Z^{-1}(z)^\good}.
\label{eq:inplace}
\end{align}
After applying the state generation map, we simply uncompute the map $\z$ to erase the register $\ket{z}$.

The algorithm begins by creating a uniform superposition over all of $z \in \F_q^{d+1}$ and applying the map \eq{inplace}.  As in the proof of \lem{inverse}, we can detect whether $z \notin \smash{R_k^\good}$, and we can postselect on the outcomes for which $z \in \smash{R_k^\good}$ with reasonable overhead, giving the state $\sum_{z \in R_k^\good} \ket{Z^{-1}(z)^\good}/|R_k^\good|^{1/2}$.
Then perform $k$ phase queries and apply the inverse of the transformation \eq{inplace}, giving the state
\begin{align}
	\frac{1}{\sqrt{|R_k^\good|}} \sum_{z \in R_k^\good} \chr(c \cdot z) \ket{Z^{-1}(z)^\good}
  \mapsto
	\frac{1}{\sqrt{|R_k^\good|}} \sum_{z \in R_k^\good} \chr(c \cdot z) \ket{z}.
\end{align}
As discussed in \sec{algdesc}, measuring this state gives $c$ with probability $|R_k^\good|/q^{d+1}$.
\end{proof}

Finally, we show how to prepare $\ket{Z^{-1}(z)^{\good}}$ and thereby give a gate-efficient quantum algorithm for polynomial interpolation with $k=d/2+1$ queries.

\begin{proof}[Proof of \thm{efficiency}\itm{eff_aboveth}: $k=d/2+1$]
We use $\ket{Z^{-1}(z)^{\good}}$ as a quantum representative of the set of solutions $Z^{-1}(z)^\good$ as described in \lem{super_rep}.  We claim that we can efficiently perform the transformation $\ket{z} \mapsto \ket{Z^{-1}(z)^{\good}}$ for a fraction $1-o(1)$ of those $z \in R_k^\good$, which in turn are a fraction $1-o(1)$ of all $z \in \F_q^{d+1}$ (by \lem{rgood_above_threshold}), giving the claimed success probability.

To prepare $\ket{Z^{-1}(z)^\good}$, we first prepare a uniform superposition over $z_{d+1} \in \F_q$ and use the procedure of \sec{efficiency_at_threshold} to compute the corresponding $(x,y)$, if it exists.  \lem{random_extra_z} shows that a fraction $(1-o(1))/k!$ of the values of $z_{d+1}$ correspond to a valid $(x,y)$, so this process can be boosted to prepare a state close to $\ket{Z^{-1}(z)^\good}$ with overhead $O(k!)$ (or with amplitude amplification, $O(\sqrt{k!})$), which in particular is independent of $q$.  We can easily uncompute $z_{d+1}$ given $(x,y)$, giving the desired transformation.
\end{proof}

\section{Open problems}
\label{sec:discussion}

In this paper, we have precisely characterized the quantum query complexity of polynomial interpolation. We conclude by briefly discussing some possible directions for future work.

In \sec{efficiency}, we gave an algorithm for the case $k = d/2 + 1$ whose gate complexity is larger than its query complexity by a factor of $k! \poly(\log q)$. This gate complexity is polynomial in $\log(q)$ but superexponential in $d$. Is it possible to give an algorithm with gate complexity only $\poly(d,\log q)$?

A natural extension of our results would be to consider the problem of learning a multivariate polynomial $f\in \FF_q[X_1,\dots,X_n]$ of degree at most $d$. Montanaro gave asymptotically optimal bounds for this problem assuming $f$ is multilinear \cite{Mon12}, but it is also natural to consider the more general case where $f$ is not necessarily multilinear. The quantum algorithm described in \sec{algdesc} can be extended to the multivariate case in a fairly straightforward manner, and we conjecture that it performs as follows.

\begin{conjecture}
	\label{conj:MultiVariate}
For any fixed positive integers $d$ and $n$, there exists a $k$-query quantum algorithm for interpolating a degree-$d$ multivariate polynomial in $n$ variables that, as $q$ grows, has success probability $1-o(1)$ provided $k>\binom{n+d}{d}/(n+1)$. 
\end{conjecture}
Note that classically one needs $\binom{n+d}{d}$ queries to solve the same problem, so our conjecture states that the quantum query complexity is smaller by a factor of $n+1$. We now discuss why computing the success probability of the quantum algorithm appears to be a difficult problem in algebraic geometry. 

Let $f\in\FF_q[X_1,\dots,X_n]$ be of degree at most $d$. 
For $j\in\NN^n$ and $x\in\FF_q^n$, we let 
\begin{equation}
x^j := \prod_{t=1}^n x_t^{j_t}. 
\end{equation}
To define the set of possible polynomials, we use the set of allowed exponents 
\begin{equation}
\cJ := \{j \in\NN^n: j_1+\cdots +j_n \leq d \}, 
\end{equation} 
with size
\begin{equation}
J := |\cJ| = \binom{n+d}{d}.
\label{eq:multexp}
\end{equation} 

We now define the function 
$\z\colon (\FF_q^n)^k\times \FF_q^k \rightarrow \FF_q^J$  by 
\begin{equation}
\z(x,y)_j = \sum_{i=1}^k y_i x_i^j
\end{equation} 
and consider its range 
\begin{equation}
\cR_k
:= \z((\FF_q^n)^k\times \FF_q^k)
\subseteq \FF_q^J. 
\label{eq:multrange}
\end{equation}
A straightforward generalization of the univariate interpolation algorithm described in \sec{algdesc} gives a multivariate interpolation algorithm with success probability $|\cR_k|/q^J$. We expect that this algorithm solves the interpolation problem with probability $1-o(1)$ using $\floor{J/(n+1)}+1$ queries.  This would be implied by the following:

\begin{conjecture}
	\label{conj:MultVar}
	With $J := \binom{n+d}{d}$ and $\cR_k$ as in \eq{multrange}, we have
	$|\cR_k| = q^{J} (1-o(1))$ provided $k>J/(n+1)$.
\end{conjecture}

Note that this holds for $n=1$ (according to \lem{rgood_above_threshold}) and also for $d=1$.  Unfortunately, the approach via exponential sums used in the proof of \lem{rgood_above_threshold} only works if $k > J/2$. Thus, while it gives a tight result for $n=1$, it appears to be inefficient for $n > 1$.

Another way to approach \conj{MultVar} is to consider the affine variety 
\begin{equation}
\cV_k\colon \z(x,y) = z
\end{equation}
in $kn + k + J$ variables $x \in (\FF_q^n)^k$, $y \in \FF_q^k$, $z \in \FF_q^J$. Clearly $|\cV_k(\FF_q)| = q^{kn+k}$. It is not hard to show that $\cV_k$ is a complete intersection and has only one absolutely irreducible component. Thus it suffices to show that for almost all specializations of $z \in \FF_q^J$, the corresponding variety $\cV_k(\vec{z})$ is absolutely irreducible; then provided $k(n+1) > J$, a version of the Lang-Weil bound~\cite{LaWe54} applies and gives the desired result. Although results of this type are known (see~\cite{ChPo16,Poon04} and references therein), unfortunately none of them seems to imply the desired statement. Nevertheless, since a generic variety is absolutely irreducible, the conjecture appears plausible.

\section*{Acknowledgments}

AMC acknowledges support from ARO (grant W911NF-12-1-0482), CIFAR, IARPA, NSF (grant 1526380), and NRO.
WvD acknowledges support from NSF (grant 1314969).
IES acknowledges support from ARC (grant DP140100118).




\begin{bibdiv}
\begin{biblist}[\normalsize]

\bib{BCD05}{inproceedings}{
      author={Bacon, Dave},
      author={Childs, Andrew~M.},
      author={van Dam, Wim},
       title={From optimal measurement to efficient quantum algorithms for the
  hidden subgroup problem over semidirect product groups},
        date={2005},
   booktitle={Proceedings of the 46th IEEE Symposium on Foundations of
  Computer Science},
       pages={469\ndash 478},
      eprint={arXiv:quant-ph/0504083},
}

\bib{BBCMW01}{article}{
      author={Beals, Robert},
      author={Buhrman, Harry},
      author={Cleve, Richard},
      author={Mosca, Michele},
      author={de~Wolf, Ronald},
       title={Quantum lower bounds by polynomials},
        date={2001},
     journal={Journal of the ACM},
      volume={48},
      number={4},
       pages={778\ndash 797},
      eprint={quant-ph/9802049},
}

\bib{BZ13}{inproceedings}{
      author={Boneh, Dan},
      author={Zhandry, Mark},
       title={Quantum-secure message authentication codes},
        date={2013},
   booktitle={Proceedings of Eurocrypt},
       pages={592\ndash 608},
}

\bib{ChPo16} {article}{
    AUTHOR = {Charles, Fran{\c{c}}ois},
     AUTHOR = {Poonen, Bjorn},
     TITLE = {Bertini irreducibility theorems over finite fields},
   JOURNAL =  {Journal of the American Mathematical Society},
    VOLUME = {29},
      YEAR = {2016},
     PAGES = {81\ndash94} 
     }

\bib{CD05}{inproceedings}{
      author={Childs, Andrew~M.},
      author={van Dam, Wim},
       title={Quantum algorithm for a generalized hidden shift problem},
        date={2007},
   booktitle={Proceedings of the 18th\ ACM-SIAM Symposium on Discrete
  Algorithms},
       pages={1225\ndash 1234},
      eprint={arXiv:quant-ph/0507190},
}

\bib{Dam98}{inproceedings}{
      author={van Dam, Wim},
       title={Quantum oracle interrogation: Getting all information for almost
  half the price},
        date={1998},
   booktitle={Proceedings of the 39th IEEE Symposium on Foundations of Computer
  Science},
       pages={362\ndash 367},
      eprint={arXiv:quant-ph/9805006},
}

\bib{DDW07}{article}{
      author={Decker, Thomas},
      author={Draisma, Jan},
      author={Wocjan, Pawel},
       title={Efficient quantum algorithm for identifying hidden polynomials},
        date={2009},
     journal={Quantum Information and Computation},
      volume={9},
      number={3},
       pages={215\ndash 230},
      eprint={arXiv:0706.1219},
}

\bib{FGGS98}{article}{
      author={Farhi, Edward},
      author={Goldstone, Jeffrey},
      author={Gutmann, Sam},
      author={Sipser, Michael},
       title={Limit on the speed of quantum computation in determining parity},
        date={1998},
     journal={Physical Review Letters},
      volume={81},
      number={24},
       pages={5442\ndash 5444},
      eprint={quant-ph/9802045},
}

\bib{GG13}{book}{
  author = {von zur Gathen, Joachim},
  author = {Gerhard, J{\"u}rgen},
  title = {Modern Computer Algebra},
  publisher = {Cambridge University Press},
  year = {2013},
}

\bib{KK11}{article}{
      author={Kane, Daniel~M.},
      author={Kutin, Samuel~A.},
       title={Quantum interpolation of polynomials},
        date={2011},
     journal={Quantum Information and Computation},
      volume={11},
      number={1},
       pages={95\ndash 103},
      eprint={arXiv:0909.5683},
}

\bib{KK90}{article}{
    AUTHOR = {Knopfmacher, Arnold},
    AUTHOR = {Knopfmacher, John},
     TITLE = {Counting polynomials with a given number of zeros in a finite
              field},
   JOURNAL = {Linear and Multilinear Algebra},
    VOLUME = {26},
      YEAR = {1990},
    NUMBER = {4},
     PAGES = {287--292},
}

\bib{LaWe54} {article}{
      author= {Lang, Serge} , 
      author= {Weil, Andr{\'e}},
      title={Number of points of varieties in finite fields},
      journal={American Journal of Mathematics},
      year={1954}
      volume={76},
       pages={819\ndash 827}
}

\bib{MP11}{article}{
      author={Meyer, David~A.},
      author={Pommersheim, James},
       title={On the uselessness of quantum queries},
        date={2011},
     journal={Theoretical Computer Science},
      volume={412},
      number={51},
       pages={7068\ndash 7074},
      eprint={arXiv:1004.1434},
}

\bib{Mon12}{article}{
  author = {Ashley Montanaro},
  title = {The quantum query complexity of learning multilinear polynomials},
  journal = {Information Processing Letters},
  volume = {112},
  number = {11},
  pages = {438--442},
  year = {2012},
  eprint = {arXiv:1105.3310},
}

\bib{PCM88}{article}{
  author = {Gerald M. Pitstick}, 
  author = {Jo{\~a}o R. Cruz},
  author = {Robert J. Mulholland},
  title = {A novel interpretation of {P}rony's method},
  journal = {Proceedings of the IEEE}, 
  volume = {76},
  number = {8},
  pages = {1052--1053}, 
  date = {1988}
}
		
\bib{Poon04} {article}{
    AUTHOR = {Poonen, Bjorn},
     TITLE = {Bertini theorems over finite fields},
   JOURNAL =  {Annals of Mathematics},
    VOLUME = {160},
      YEAR = {2004},
     PAGES = {1099\ndash1127}
}

\bib{RSV02}{article}{
  author = {Jaikumar Radhakrishnan},
  author = {Pranab Sen},
  author = {S. Venkatesh},
  title = {The quantum complexity of set membership},
  journal = {Algorithmica},
  pages = {462\ndash 479},
  date = {2002},
  eprint = {arXiv:quant-ph/0007021},
}

\bib{Sha79}{article}{
      author={Shamir, Adi},
       title={How to share a secret},
        date={1979},
     journal={Communications of the ACM},
      volume={22},
      number={11},
       pages={612\ndash 613},
}

\end{biblist}
\end{bibdiv}
\end{document}